\documentclass[12pt]{article}
\usepackage{amsmath}
\usepackage{amsthm}
\usepackage{graphicx}
\usepackage{natbib}
\usepackage{url} 

\usepackage{bm}
\usepackage{xcolor}
\usepackage{booktabs}

\newcommand{\E}{\mathrm{E}}

\newcommand{\N}{\mathcal{N}}

\newcommand{\tr}{\text{tr}}

\newcommand{\blind}{0}

\newtheorem{theorem}{Theorem}[section]

\newtheorem{lemma}[theorem]{Lemma}

\begin{document}

\def\spacingset#1{\renewcommand{\baselinestretch}%
{#1}\small\normalsize} \spacingset{1}


\if0\blind
{
  \title{\bf An efficient and accurate approximation to the distributions of quadratic forms of Gaussian variables}
  \author{Hong Zhang, Judong Shen\hspace{.2cm}\\
    Biostatistics and Research Decision Sciences, MRL, Merck \& Co., Inc.\\
    and \\
    Zheyang Wu \\
Department of Mathematical Sciences, Worcester Polytechnic Institute}
  \maketitle
} \fi

\if1\blind
{
  \bigskip
  \bigskip
  \bigskip
  \begin{center}
    {\LARGE\bf Title}
\end{center}
  \medskip
} \fi

\bigskip
\begin{abstract}
In computational and applied statistics, it is of great interest to get fast and accurate calculation for the distributions of the quadratic forms of Gaussian random variables. This paper presents a novel approximation strategy that contains two developments. First, we propose a faster numerical procedure in computing the moments of the quadratic forms. Second, we establish a general moment-matching framework for distribution approximation, which covers existing approximation methods for the distributions of the quadratic forms of Gaussian variables. Under this framework, a novel moment-ratio method (MR) is proposed to match the ratio of skewness and kurtosis based on the gamma distribution. Our extensive simulations show that 1) MR is almost as accurate as the exact distribution calculation and is much more efficient; 2) comparing with existing approximation methods, MR significantly improves the accuracy of approximating far right tail probabilities. The proposed method has wide applications. For example, it is a better choice than existing methods for facilitating hypothesis testing in big data analysis, where efficient and accurate calculation of very small $p$-values is desired. An R package \emph{Qapprox} that implements related methods is available on CRAN.
\end{abstract}

\noindent%
\vfill

\newpage
\spacingset{1.5} 

\section{Introduction}\label{sec:intro}
Quadratic forms of Gaussian variables appear in many statistical applications such as hypothesis testing and statistical power calculation. 
For instance, recent developments in statistical genetics often adopt the quadratic forms as test statistics to detect associations between genetic variants and complex phenotypes. Well-known examples include, to name a few, the kernel machine based score statistic for genetic pathway analysis \citep{liu2007semiparametric}, the sequence kernel association test statistic (SKAT) \citep{wu2011rare} and the sum of powered score test statistic (SPU) \citep{pan2014powerful} for rare-variant analysis. 
These new applications of analyzing big biological data require accurate calculation of small $p$-values. This requirement poses challenges to the calculation of distribution at the far-right tail. For example, in a typical gene-based whole-genome association study of about 20,000 human genes, the significance threshold under Bonferroni correction is at the level of $\alpha = 0.05/20000=2.5\times10^{-6}$.

In this paper we study efficient and accurate algorithms to calculate the distributions of the quadratic forms of Gaussian variables. Specifically, for a vector of $n$ Gaussian variables, $X\sim \N(\mu,\Sigma)$ with mean vector $\mu$ and covariance matrix $\Sigma$, its quadratic form is
\begin{align}
Q = X^\prime AX, 
\end{align}
where $A$ is an $n\times n$ positive semi-definite matrix. We are interested in calculating the distribution of $Q$, or, equivalently its right tail probability:
\begin{equation}
\label{equ.tail}
P(Q>q), 
\end{equation}
where $q\geq0$ is an observed value of $Q$.

First, note that the distribution of $Q$ can be calculated in an exact way. 
Consider a `decorrelation' of $X$ by the inverse of the Cholesky of $\Sigma$, $U=(\Sigma^{1/2})^{-1}X\sim \N((\Sigma^{1/2})^{-1}\mu, I)$. 
We can write $Q=U^\prime M U$, where $M = (\Sigma^{1/2})^\prime A\Sigma^{1/2}$. Denote $\Lambda$ the diagonal matrix of eigenvalues of $M$ and $P$ the orthonormal matrix such that $M=P\Lambda P^\prime$, we have
\begin{align}
\label{equ.Q}
Q\overset{d}{=}(Z+\tilde{\mu})^\prime \Lambda(Z+\tilde{\mu}),
\end{align}
where $Z$ is a vector of $n$ independent standard normal variables and $\tilde{\mu}=P^\prime(\Sigma^{1/2})^{-1}\mu$.
The formula (\ref{equ.Q}) indicates that the distribution of $Q$ is the same as a weighted sum of independent, potentially non-central, chi-squared random variables. Thus, the cumulative distribution function (CDF) of $Q$ can be obtained by inverting its characteristic function numerically~\citep{imhof1961computing, davies1980algorithm}. We call this calculation approach the exact method since it is in theory accurate except potential errors in the numerical procedures. 
However, the exact approach is time-consuming for two reasons: 1) it requires eigendecomposition with a computation cost in the order of $O(n^3)$, which is heavy when $n$ is large in big data analysis; 2) the algorithm's speed is sensitive to the number of positive eigenvalues and the point at which the CDF is to be evaluated~\citep{davies1980algorithm}. 

To speed up the computation, various approximation methods have been proposed based on the moments of $Q$. 
In particular, the Satterthwaite-Welch method (SW)~\citep{welch1938significance, satterthwaite1946approximate} matches the first two moments of $Q$ with a gamma variable. The Hall-Buckley-Eagleson approximation (HBE)~\citep{hall1983chi, buckley1988approximation} matches the skewness of $Q$ with a chi-squared variable while adjusting for the mean and the variance. Wood's $F$ approximation (Wood)~\citep{wood1989f} matches the first three moments with a three-parameter $F$ variable. Liu-Tang-Zhang method (LTZ)~\citep{liu2009new} tries to match both skewness and kurtosis of $Q$ with a non-central chi-squared variable while adjusting for the mean and the variance at the same time. A more comprehensive literature review can be found in~\citep{bodenham2016comparison}. 
In the applications to hypothesis testing of big data, these methods are lack of desired accuracy for controlling the type I error rate at small $\alpha$ levels (see Figure \ref{fig.tie_box_mu}). Furthermore, to obtain the moments of $Q$, these methods typically use eigenvalues of $A\Sigma$ or the trace of $(A\Sigma)^k$, $k=1,2,...$, which can be computationally intensive when $n$ is large.

To address these two issues, we first propose a numerical procedure that calculates the moments of $Q$ more efficiently in Section~\ref{sec.moments}. Then, in Section~\ref{sec.framework}, we describe a general moments-matching framework, which 
allows flexibly incorporating the information of high moments to improve accuracy without increasing computational difficulties.
Within this framework we propose a novel approximation method in Section~\ref{sec.MR}, that significantly improves the accuracy in computing the distribution of $Q$, especially for the far right tail probability. The computation time and type I error comparison results from extensive simulation study are presented in Section~\ref{sec.simu}. We conclude this paper with final remarks in Section~\ref{sec.conc}. 


\section{Methods}
\subsection{Computing the moments of $Q$}\label{sec.moments}
In order to approximate the distribution of $Q$, the first step is to calculate its moments. This step takes the dominate computational time in the approximation methods. We propose a strategy to speed up this process, which we didn't see in the relevant $Q$ distribution approximation literature yet but can be applied into these methods to improve computational efficiency. 

For a general quadratic form $Q$, the $k$th cumulant, $c_k$, is given by \citep{johnson1995continuous}, 
\begin{align}
\label{equ.c_eigen}
c_k &= 2^{k-1}(k-1)!\left(\tr \left(\Lambda^k\right) + k \tilde{\mu}' \Lambda^k\tilde{\mu}\right).
\end{align}
Subsequently, the mean $\mu_Q$, variance $\sigma_Q^2$, skewness $\gamma_Q$ and kurtosis $\kappa_Q$ can be calculated through $c_k$, $k=1,2,3,4$: 
\begin{equation}
\label{equ.moments_c}
\mu_Q=c_1, \quad{}\sigma_Q^2=c_2, \quad{} \gamma_Q=\frac{c_3}{c_2^{3/2}}, \quad{} \kappa_Q=\frac{c_4}{c_2^2}+3, 
\end{equation}


The formula in (\ref{equ.c_eigen}) requires calculating the eigenvalues, which has a computation cost in the order of $O(n^3)$ in practice. To improve the computation, \cite{liu2009new} proposed a ``trace" approach by rewriting (\ref{equ.c_eigen}) to be
\begin{align}
\label{equ.c_tr}
c_k  &= 2^{k-1}(k-1)!\left(\tr \left((A\Sigma)^k\right) + k \mu' (A\Sigma)^{k-1}A\mu\right).
\end{align}
Although formula (\ref{equ.c_tr}) does not need eigenvalues explicitly, the naive implementation of this approach, assuming $A\Sigma$ is known in advance, will require $k-1$ matrix multiplication. This approach could still be computationally expensive especially when $k$ is large. For example, it requires at least three matrix multiplications to get $c_1$ to $c_4$, which is 
not much more efficient than the eigenvalue approach (formula (\ref{equ.c_eigen})) as evidenced by Figure~\ref{fig.time_moments}.

We can improve the computational efficiency by not explicitly computing every matrix multiplication (see Lemma \ref{thm.moments}). The improvement is significant when the Gaussian variables are centered, i.e., $\mu\equiv0$. 

\begin{lemma}
\label{thm.moments}
Let $B$ be an $n\times n$ matrix. If $B^{k-1}$ and $B^{k}$, $k\geq 1$, are known, then we can compute $\tr (B^{2k-1})$ and $\tr(B^{2k})$ in $O(n^2)$ time.
\end{lemma}
\begin{proof}
Through matrix multiplication we have $\tr (B^{2k-1}) = \sum_{i=1}^n\sum_{j=1}^n(B^{k-1})_{ij}(B^k)_{ji}$ and  
$\tr (B^{2k}) = \sum_{i=1}^n\sum_{j=1}^n(B^k)_{ij}(B^k)_{ji}$. 
\end{proof}

Lemma~\ref{thm.moments} says that we can compute the trace of $B^k$ by roughly $k/2$ (instead of $k-1$) matrix multiplications. In particular, it only takes a single matrix multiplication to get the first four moments. Thus, by using Lemma~\ref{thm.moments} the computation cost of getting the first four moments is about one-third of the naive trace approach. If only mean and variance are needed, as in the SW method, the computational efficiency is further improved since no matrix multiplication is needed by Lemma~\ref{thm.moments}. As will be shown later, SW is not very accurate for very small $p$-values. However, it is still adequate when the significance level is not too stringent, e.g., $\alpha=0.05$. Because of its fast speed, it might be useful as a screening method to speed up the overall data analysis. For example, in the genome-wide association studies most genes are not trait associated and thus their $p$-values are expected to be relatively large. The SW can be utilized to remove these genes before applying more accurate but computationally more intensive methods to calculate small $p$-values.

\subsection{A general framework for moment-matching method}\label{sec.framework}
Here we describe a general moment-matching framework that covers the literature approximation methods in Introduction. 
Ideally, engaging more moments and more parameters could allow more flexible distribution model and thus provide more accuracy in general. However, the resulting matching equations would also be more difficult to solve, and sometimes a solution does not exist. 
To ease this trade-off, our general  moment-matching framework can engage more moments without adding too much computational difficulty.

Specifically, let $Y$ 
be a random variable with known CDF $F_Y(y|\bm{\theta})$, where $\bm{\theta}=(\theta_1,...,\theta_m)$ denotes $m$ distribution parameters. In this framework, we match the moments of $Q$ with the moments of a linear transformation of $Y$: 
\begin{equation}
\label{equ.framework}
T=aY+b. 
\end{equation}
Note that the standardized moments of $T$ are always equal to those of $Y$, which are functions of $\bm{\theta}$ but not of $a$ and $b$. That is, for $k\geq 1$,
\begin{equation*}
\tilde{\mu}_{k,Y}(\bm{\theta}) :=\E\left(\frac{Y-\mu_Y}{\sigma_Y}\right)^k = \E\left(\frac{T-\mu_T}{\sigma_T}\right)^k =: \tilde{\mu}_{k,T}(\bm{\theta}).
\end{equation*}

The advantage of matching $Q$ with $T$, instead of directly matching $Q$ with $Y$, is that it has a potential to gain two ``free" parameters $a$ and $b$ corresponding to the mean and variance. 
In other words, the means and variances of $Q$ and $T$ are automatically matched as long as $\bm{\theta}$ is available. Therefore, $\bm{\theta}$ could be determined by matching higher ($k\geq 3$) standardized moments (e.g., skewness $\gamma=\tilde{\mu}_3$ and kurtosis $\kappa=\tilde{\mu}_4$), which provides extra flexibility of distribution especially at the tail. 

Specifically, by matching the mean and variance of $Q$ and $T$ we have
\begin{equation}
\label{equ.ab}
a^* = \frac{\sigma_Q}{\sigma_Y(\bm{\theta})}, \quad{} b^* = \mu_Q - \frac{\sigma_Q}{\sigma_Y(\bm{\theta})}\mu_Y(\bm{\theta}).
\end{equation}
Now, $\bm{\theta}$ could be determined by solving proper equations regarding the standardized higher ($k\geq 3$) moments of $Q$ and $Y$. We denote such equations in the following generic way: 
\begin{equation}
\label{equ.higher}
g_j(\tilde{\mu}_{k,Y}(\bm{\theta}),\tilde{\mu}_{k,Q}(\bm{\theta}); k=3, 4, ...) = 0, \quad{} j=1,2,...,m,
\end{equation}
where the $g_j$ functions of these standardized higher moments need to be designed properly in order to get solutions $\bm{\theta}^*$. Finally, with the determined $a^*$, $b^*$ and $\bm{\theta}^*$ the right-tail probability of $Q$ is approximated by
\begin{equation}
\label{equ.approx}
P(Q>q)\approx P(T>q) = P\left(Y>\frac{q-b^*}{a^*}\right)=1-F_Y\left(\frac{\sigma_Y}{\sigma_Q} (q-\mu_Q) + \mu_Y |\bm{\theta}^*\right).
\end{equation}

The existing approximation methods for $Q$ can fit into this general framework in (\ref{equ.framework}) -- (\ref{equ.approx}). 
When choosing $Y$ a gamma random variable $G(\alpha,1)$ with shape parameter $\alpha$ and scale parameter $1$, the Satterthwaite-Welch method is equivalent to solving $\bm{\theta} = \alpha$ by fixing $b^*=0$ in (\ref{equ.ab}) without engaging higher moment information in (\ref{equ.higher}). 
Hall-Buckley-Eagleson method also uses the gamma variable; it is equivalent to solving $\alpha$ by matching the skewness in (\ref{equ.higher}). 
Choosing $Y$ a non-central chi-squared variable $\chi_d^2(\delta)$ with degrees of freedom $d$ and noncentrality $\delta$, Liu-Tang-Zhang method is equivalent to solving $\bm{\theta} = (d, \delta)$ by matching the skewness and kurtosis in (\ref{equ.higher}). 
Letting $Y$ follow an $F$ distribution with degrees of freedom $d_1$ and $d_2$, Wood's F approximation is equivalent to solving $\bm{\theta} = (d_1, d_2)$ by fixing $b^*=0$ in (\ref{equ.ab}) and matching the skewness in (\ref{equ.higher}). 

\subsection{Strategies for matching higher moments}\label{sec.MR}
Solving $\bm{\theta}$ by matching  standardized higher moments in (\ref{equ.higher}) is not trivial. Depending on the distribution of $Y$, due to restrictions on the domains of these moments and their interdependence, the solution to equations in (\ref{equ.higher}) often do not exist for many seemingly obvious choices of the $g$ functions.  For example, one design for (\ref{equ.higher}) is to match skewness and kurtosis at the same time, such as described in the Liu-Tang-Zhang method ~\citep{liu2009new}. 
The authors provided a nice result on the necessary and sufficient condition, i.e., $c_3^2>c_4c_2$, for the solution of $\bm{\theta} = (d, \delta)$ exists. For centered Gaussian variables, the condition becomes $\left(\sum_{i=1}^n\lambda_i^3\right)^2>\left(\sum_{i=1}^n\lambda_i^2\right)\left(\sum_{i=1}^n\lambda_i^4\right)$. In this case, however, a general inequality given in Lemma~\ref{thm.holder} indicates that such condition will never be met. This is somewhat surprising because $Q$ follows a chi-square distribution when $\Sigma=I$. However,  when $\Sigma\neq I$, it is impossible to find a chi-square distribution that matches $Q$'s skewness and kurtosis by adding a non-centrality parameter.
Even for non-central Gaussian variables, the solution still often does not exist 
(see \cite{duchesne2010computing, bodenham2016comparison} and the simulations in Section \ref{sec.simu}). When the solution does not exist, 
Liu-Tang-Zhang method steps back to matching only the skewness, which is exactly the Hall-Buckley-Eagleson method. We can also modify the method to matching the kurtosis instead, which could show some improved results in some scenarios. However, in either way it will lose the information of the unmatched moment. 

\begin{lemma}[Littlewood's inequality]
\label{thm.holder}
Define $S_a= \sum_{i=1}^n\lambda_i^a$, $a>0$, $\lambda_i\geq 0$, $i=1,...,n$. Suppose $0< a< b< c$, then $S_b^{c-a}\leq S_a^{c-b}S_c^{b-a}$. The equality holds if and only if $\lambda_1=\lambda_2=\cdots=\lambda_n$.
\end{lemma}
\begin{proof}
The lemma follows a classic inequality that can be found in literature (cf. \cite{hardy1934inequalities}, p.28, Th. 18). Here we show it is equivalent to the Hölder's inequality. 

Write $b = ka + (1-k)c$, $0<k<1$. Then $c-b=k(c-a)$, $b-a=(1-k)(c-a)$. The inequality is equivalent to 
\begin{align*}
S_b^{c-a}\leq S_a^{k(c-a)}S_c^{(1-k)(c-a)} \Longleftrightarrow \sum_{i=1}^n\lambda_i^{ka}\lambda_i^{(1-k)c}\leq  \left(\sum_{i=1}^n\lambda_i^a\right)^{k}\left( \sum_{i=1}^n\lambda_i^a\right)^{1-k}.
\end{align*}
Define $p=1/k$, $q=1/(1-k)$, $x_i=\lambda_i^{a/p}$, $y_i=\lambda_i^{c/q}$, the above inequality becomes
\begin{align*}
\sum_{i=1}^nx_iy_i\leq  \left(\sum_{i=1}^nx_i^p\right)^{1/p}\left( \sum_{i=1}^ny_i^q\right)^{1/q},
\end{align*}
which is the Hölder's inequality. The equality holds if and only if $x_i$ and $y_i$, $i=1,...,n$, are proportional which is equivalent to $\lambda_i$, $i=1,...,n$, are equal to each other. 
\end{proof}


In order to overcome the difficulty of matching both skewness and kurtosis, while at the same time still keeping the information from both, we propose the following solution. First, we choose the gamma distribution model, i.e., $Y \sim G(\alpha, 1)$. Here we fix the gamma scale parameter to be 1 because it is redundant with the coefficient parameter $a$ in (\ref{equ.framework}) and thus does not affect the calculation in (\ref{equ.approx}). Second, we propose two types of $g$ functions given below for equation (\ref{equ.higher}), which incorporate the information of both skewness and kurtosis. Comparing with the Liu-Tang-Zhang method, the solution of our matching equation always exists. 
At the same time, it can be more accurate in approximating the right tail probability in (\ref{equ.tail}) at relatively large threshold $q$. 

Specifically, the first method is called the moment-ratio matching (MR) method.
With one parameter $\alpha$, we match the ratio between skewness and kurtosis. Define 
\begin{equation*}
g_1(\tilde{\mu}_{3,Y},\tilde{\mu}_{3,Q},\tilde{\mu}_{4,Y},\tilde{\mu}_{4,Q})  = \frac{\tilde{\mu}_{3,Y}}{\tilde{\mu}_{4,Y}-3}-\frac{\tilde{\mu}_{3,Q}}{\tilde{\mu}_{4,Q}-3}.
\end{equation*}
The solution of equation (\ref{equ.higher}) always exists and has a simple closed form: $\alpha^*=\frac{9\tilde{\mu}_{3,Q}^2}{(\tilde{\mu}_{4,Q}-3)^2}$. 

The second method is called the minimized matching error (ME) method.
That is, we choose $\alpha$ to minimize the Euclidean distance between $(\tilde{\mu}_{3,Y}, \tilde{\mu}_{4,Y})$ and $(\tilde{\mu}_{3,Q}, \tilde{\mu}_{4,Q})$. For that, the corresponding $g$ function is the first-order derivative of the distance as a function of $\alpha$, 
\begin{equation}
\label{equ.me}
g_1(\tilde{\mu}_{3,Y},\tilde{\mu}_{3,Q},\tilde{\mu}_{4,Y},\tilde{\mu}_{4,Q})  = \frac{\partial}{\partial\alpha}\left[(\tilde{\mu}_{3,Y}-\tilde{\mu}_{3,Q})^2+(\tilde{\mu}_{4,Y}-\tilde{\mu}_{4,Q})^2\right].
\end{equation}
Accordingly, equation~(\ref{equ.higher}) becomes 
\begin{equation}
\label{equ.me_cubic}
\tilde{\mu}_{3,Q}\alpha^{3/2} -2(10-3\tilde{\mu}_{4,Q})\alpha - 36 = 0. 
\end{equation}
Since $\tilde{\mu}_{3,Q}>0$, it is straightforward to show that there exists one and only one real root $\alpha^*>0$.

\section{Simulation results}\label{sec.simu}

\subsection{Computation time comparison}
In this section, we compare the computation time of the proposed moment-ratio (MR) matching method with other existing methods. The comparison is two-fold. First, we examine the computation efficiency of different methods for obtaining the eigenvalues or the moment estimates of $Q$. Second, given the eigenvalues and moment estimates, we compare various methods for calculating the final right-tail probability based on different distributions.

We first consider the correlation matrix $\Sigma$ to be a polynomial decaying matrix, i.e., $(\Sigma)_{ij}=1/|i-j|$, $1\leq i,j\leq n$, $n=100,200,...,5,000$. The eigenvalue decomposition (formula (\ref{equ.c_eigen})), the naive trace method (formula (\ref{equ.moments_c})) and the proposed moment computation method (Lemma~\ref{thm.moments}) were performed 1,000 times and the average run times were summarized in the left panel of Figure~\ref{fig.time_moments}. Next, for each $n$, we simulate $50,000$ quadratic form of centered Gaussian variables, $Q=X^\prime X$, with correlation matrix $\Sigma$. The exact Davies method, the LTZ method (based on chi-square distribution), the HBE method and the proposed MR method (gamma distribution) and Wood method (F distribution) were used to calculate the $p$-values assuming the eigenvalues of $\Sigma$ are known. The total runtimes were recorded and summarized in the right panel of Figure~\ref{fig.time_moments}. For the second comparison, the Exact and LTZ methods were implemented by a widely used \emph{R} package \emph{CompQuadForm}. All other methods were implemented by the authors. All computations were completed on an Intel Core-i5 8350U processor at 1.70 GHz with 16 GB of RAM, running Microsoft R Open 3.5.0. 

As Figure~\ref{fig.time_moments} shows, the computation time of the trace approach and the eigenvalue approach is about the same but the proposed moment calculation method can save the computation time by 1/2 to 2/3. Compared to the exact calculation by Davies method, the approximation methods for calculating the final right-tail probability (based on chi-square or gamma or F distribution) could be hundreds of times more efficient. Moreover, unlike Davies method, the computation costs of the approximation methods do not increase as the dimension $n$ increases. Meanwhile, we also note that the overall computation time is dominated by the calculation of the eigenvalues/moments.

\begin{figure}
\includegraphics[width=0.5\textwidth]{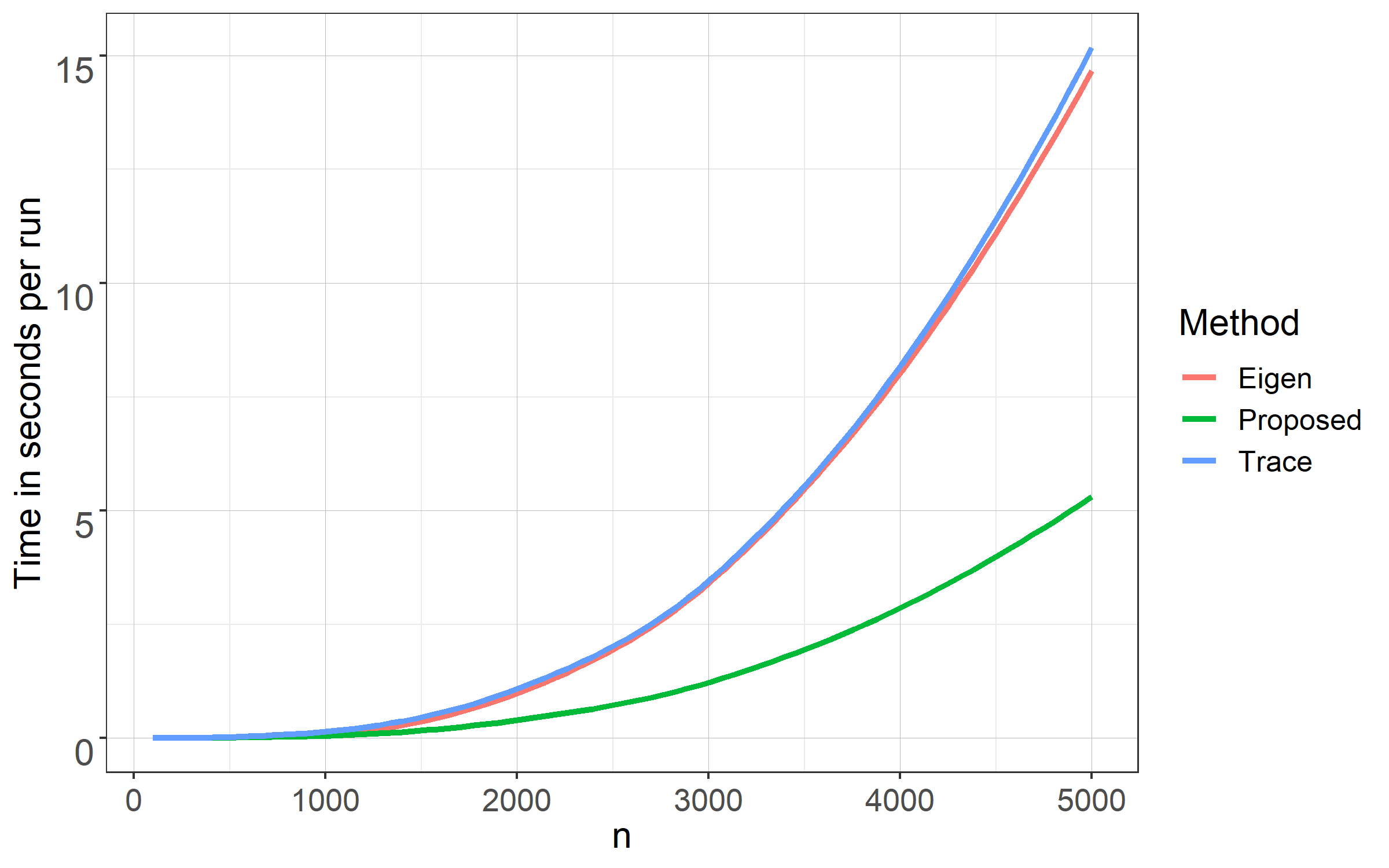}
\includegraphics[width=0.5\textwidth]{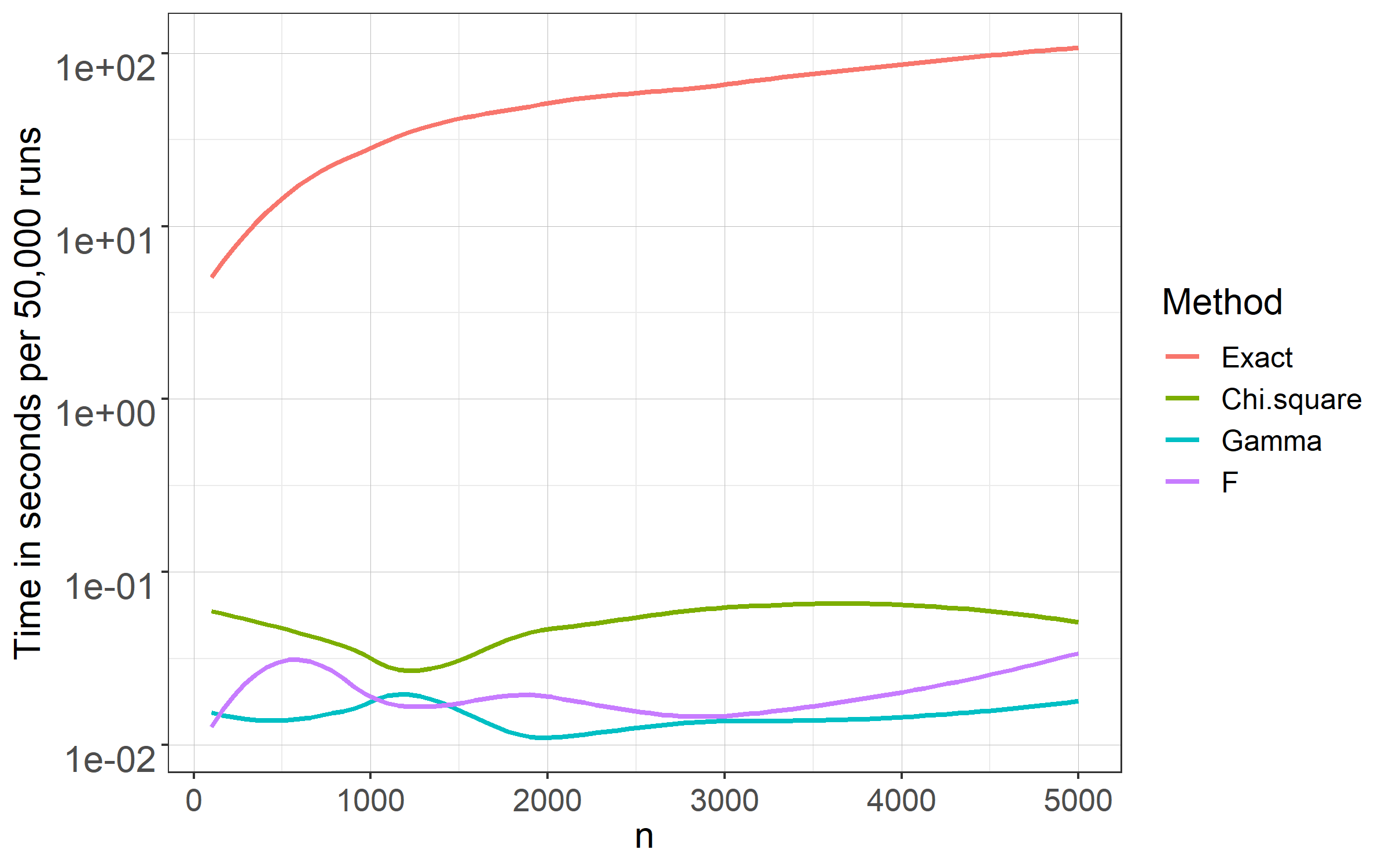}
\caption{Computation time comparison. Left: eigenvalues/moments computation comparison. Eigen: eigenvalue decomposition (formula (\ref{equ.c_eigen})). Proposed: the proposed moment computation method Lemma~\ref{thm.moments}. Trace: naive trace method (formula (\ref{equ.moments_c})). Right: distribution computation comparison per $50,000$ runs. Exact: Davies method. $n$ is the dimension of the correlation matrix. 
}
\label{fig.time_moments}
\end{figure}

\subsection{Type I error comparison}
In this section we systematically evaluate the accuracies of the proposed and existing $Q$ distribution approximation methods in the context of hypothesis testing. Specifically, consider $Q$ as the test statistic and $q$ as its observed value. The right-tail probability $P(Q>q)$ can be regarded as the $p$-value for testing the null hypothesis $H_0 : \mu_X=\mu$. 
In the simulation study, we fix $A=I$ and vary the correlation matrix $\Sigma$ (since $A$ can be considered as a factor that varies the correlation matrix $\Sigma$). 
For a given mean vector $\mu$ and correlation matrix $\Sigma$, $2\times 10^7$ of random Gaussian vectors $X\sim \N(\mu,\Sigma)$ were generated, each led to an observed $q$. 
Consequently, each approximation method was applied to generate $2\times 10^7$ $p$-values. The empirical type I error rate for a given approximation method is defined as the proportion of rejections under $H_0$, i.e., the proportion of $p$-values less than a nominal level $\alpha$. 
A good approximation method would have the empirical type I error rate  
being close to the nominal $\alpha$. If a ratio between the empirical type I error rate and the nominal $\alpha$ is larger than 1, it indicates the approximation method leads to a liberal test (i.e., rejecting more than appropriate); a ratio less than 1 indicates a conservative test. 

Motivated by genetic data, we examined a wide variety of covariance patterns that are potentially with block structures. Consider the $n\times n$ covariance matrix 
\[
\Sigma = 
\begin{bmatrix}
\Sigma_{11} & \Sigma_{12} \\
\Sigma'_{12} & \Sigma_{22} 
\end{bmatrix},
\]
where each of the blocks is of size $(n/2) \times (n/2)$. Equal correlation matrix $E_k$ and polynomially decaying correlation matrix $D_k$ were used to define $\Sigma$ or its blocks:
\begin{align}
\label{equ.equ-decay-matrix}
E_k(\rho) &: E_k(i,j)=\rho, \quad{} 1\leq i\neq j\leq k \text{ and }  0\leq\rho<1, \\
\label{equ.B}
D_k(\phi) &: D_k(i,j)=1/|i-j|^\phi, \quad{} 1\leq i\neq j\leq k \text{ and } \phi>0.
\end{align} 
Table~\ref{tbl.sigma} summarizes a total of 12 correlation types for $\Sigma$. 
Multiple values of the matrix parameters $\rho=0.9,0.5,0.1$ and $\phi=0.2, 1, 3$ were considered to model the strong, moderate and weak correlations, respectively. The dimensions $n=10$, $50$, $100$ and $500$ were considered. 
We also tested three configurations of the mean vector: a central Gaussian case $\mu=\mu_1\equiv 0$, and two non-central Gaussian cases $\mu=\mu_2\equiv 1$ and $\mu=\mu_3=(\underbrace{1,...,1}_{n/2},\underbrace{0,...,0}_{n/2})$. Our simulation settings are pretty extensive; altogether there are 432 different combinations of the mean $\mu$ and correlation matrix $\Sigma$ for each given $n$. 

\begin{table}[]
\centering
\caption{The correlation matrix $\Sigma$ used in the simulations based on the correlation models in (\ref{equ.equ-decay-matrix}) and (\ref{equ.B}). Three types: I (upper left correlations with $\Sigma_{22} = I$); II (block-wise correlations); III (the whole correlation matrix follows the models). $\Sigma_{12}=0$.}
\label{tbl.sigma}
\begin{tabular}{@{}llll@{}}
\toprule
Type                 & I  & II & III  \\ \midrule
Equal($\rho$)      & \begin{tabular}[c]{@{}l@{}}$\Sigma_{11}= E_{n/2}(\rho)$\end{tabular}           & \begin{tabular}[c]{@{}l@{}}$\Sigma_{11} = \Sigma_{22}= E_{n/2}(\rho)$ \end{tabular}                    & \begin{tabular}[c]{@{}l@{}}$\Sigma= E_{n}(\rho)$\end{tabular}                \\
Poly($\phi$)     & \begin{tabular}[c]{@{}l@{}}$\Sigma_{11}= D_{n/2}(\phi)$\end{tabular}              & \begin{tabular}[c]{@{}l@{}}$\Sigma_{11} = \Sigma_{22}= D_{n/2}(\phi)$\end{tabular}                     & \begin{tabular}[c]{@{}l@{}}$\Sigma= D_{n}(\phi)$\end{tabular}               \\
Inv-Equal($\rho$)*  & \begin{tabular}[c]{@{}l@{}}$\Sigma_{11}= E^{-1}_{n/2}(\rho)$\end{tabular}             & \begin{tabular}[c]{@{}l@{}}$\Sigma_{11} = \Sigma_{22}= E^{-1}_{n/2}(\rho)$\end{tabular}                    & \begin{tabular}[c]{@{}l@{}}$\Sigma= E^{-1}_{n}(\rho)$\end{tabular}                \\
Inv-Poly($\phi$)* & \begin{tabular}[c]{@{}l@{}}$\Sigma_{11}= D^{-1}_{n/2}(\phi)$\end{tabular}              & \begin{tabular}[c]{@{}l@{}}$\Sigma_{11} = \Sigma_{22}= D^{-1}_{n/2}(\phi)$\end{tabular}                    & \begin{tabular}[c]{@{}l@{}}$\Sigma= D^{-1}_{n}(\phi)$\end{tabular}  \\ \bottomrule
\multicolumn{4}{l}{\footnotesize{*$\Sigma$ is standardized to become a correlation matrix.}}
\end{tabular}
\end{table}

The approximation methods to be compared are 1) MR: the proposed moment-ratio matching method; 2) ME: the proposed minimized matching error method; 3) SW: Satterthwaite-Welch method; 4) HBE: Hall-Buckley-Eagleson method; 5) Wood: Wood's $F$ approximation; 6) LTZ: Liu-Tang-Zhang method; 7) LTZ4: modified Liu-Tang-Zhang method (i.e., when skewness and kurtosis could not be matched simultaneously, we chose to match kurtosis instead of skewness); and 8) Exact: the Davies method that inverts the characteristic function. We used R package \emph{CompQuadForm} for the LTZ method and the Exact method. Other methods were implemented by the authors. 

The boxplots in Figure~\ref{fig.tie_box_005001} show the comparison of all eight methods at relatively larger nominal $\alpha$ levels, $\alpha = 0.05$ and $0.01$. Each box consists of multiple empirical type I errors under different combinations of the correlation matrix $\Sigma$ and mean vector $\mu$. It is clear that there exist many scenarios in which the Wood method is ultra-conservative while the SW method could inflate the type I errors at $\alpha = 0.01$. The other approximation methods seem to perform well as their empirical type I error rates are within $(0.8, 1.1)$ of the nominal levels. Unsurprisingly, the Exact method is the most accurate method with the least amount of variation around the nominal levels.

Next we consider smaller nominal levels, $\alpha=1\times 10^{-4}$ and $2.5\times 10^{-6}$. The comparison results are summarized in Figure~\ref{fig.tie_box_mu} as boxplots stratified by the mean vector $\mu$. Each box represents empirical type I errors under different correlation matrices. The SW and Wood methods were dropped because their empirical type I errors were further away from $\alpha$. When $\mu=\mu_1\equiv 0$ (the central Gaussian case), as expected, the exact method was accurate at both levels. The proposed moment ratio matching method (MR) also controlled the type I error rates very well. It only varied slightly more than the exact method when $\alpha= 2.5\times 10^{-6}$ and is slightly conservative when $\alpha=10^{-4}$. The minimized matching error method (ME) was slightly more inflated compared to MR. The Liu-Tang-Zhang (LTZ) method performed the same as the Hall-Buckley-Eagleson (HBE) method because it cannot match both skewness and kurtosis when $\mu\equiv 0$, as proven by Lemma \ref{thm.holder}. Therefore, in this case it reduced to matching the skewness only. Both methods seemed to inflate the type I errors in multiple occasions across the dimension $n$. The modified Liu-Tang-Zhang (LTZ4) method seemed to perform better than the original version but still with considerable amount of inflations. When $\mu=\mu_2$ or $\mu_3$ (the non-central Gaussian case), the inflation of LTZ's and LTZ4's type I errors still existed but were greatly reduced because now in various scenarios both skewness and kurtosis could be matched.  The proposed MR method still controlled type I errors well with small variation around $\alpha$.

To further demonstrate the impact of unmatched moments on the type I errors, we present the boxplots stratified by whether the skewness and kurtosis can be matched (`unmatched' vs. `matched') in the LTZ method. Figure~\ref{fig.tie_box_LTZ} shows that when these two moments were matched, the performance of LTZ method is on par with the MR method and the exact method. However, if they are not matched, choosing either skewness or kurtosis  would inflate the type I errors while the MR method would not.

\begin{figure}
\includegraphics[width=1\textwidth]{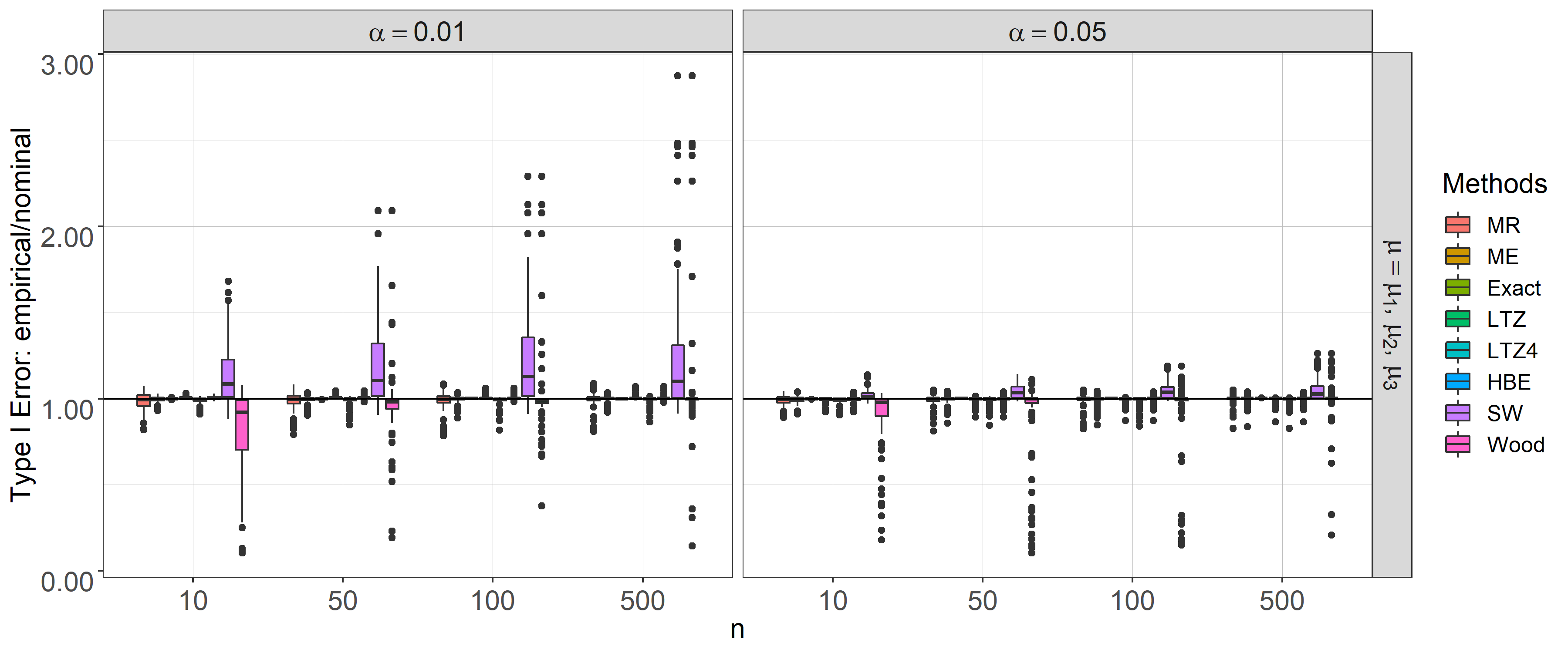}
\caption{Boxplots of the ratios between empirical type I error rates and the nominal levels, $\alpha = 0.01$ and $0.05$.  MR: moment-ratio matching method. ME: minimized matching error method. Exact: the Davies method. LTZ: Liu-Tang-Zhang method. LTZ4: modified Liu-Tang-Zhang method. HBE: Hall-Buckley-Eagleson method. SW: Satterthwaite-Welch method. Wood: Wood's $F$ approximation.} 
\label{fig.tie_box_005001}
\end{figure}

\begin{figure}
\includegraphics[width=1\textwidth]{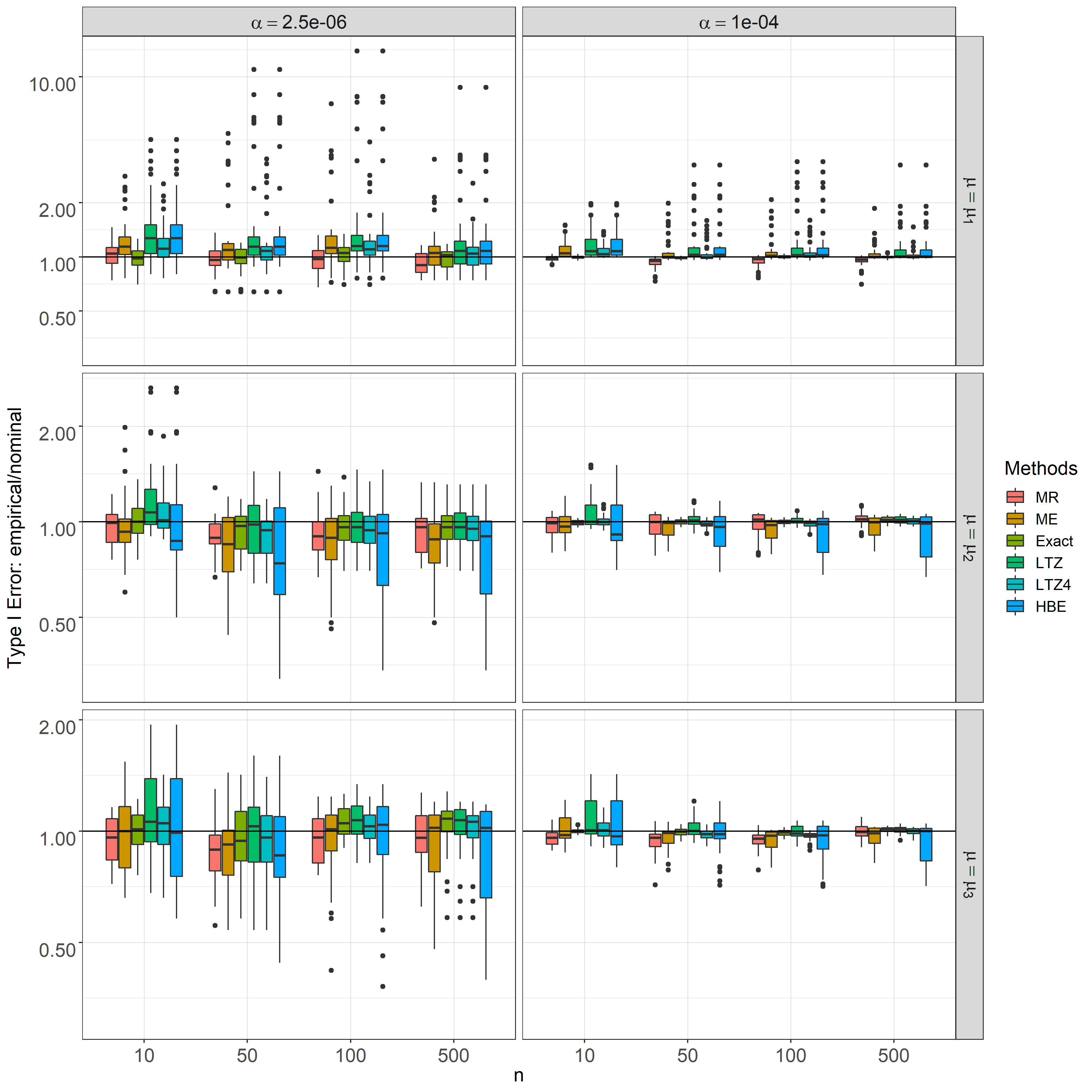}
\caption{Boxplots of the ratios between empirical type I error rates and the nominal levels, $\alpha = 2.5\times 10^{-6}$ and $1\times 10^{-4}$, stratified by the mean vector $\mu_1\equiv 0$, $\mu_2\equiv 1$ and $\mu_3=(1,...,1,0,...,0)$.  MR: moment-ratio matching method. ME: minimized matching error method. Exact: the Davies method. LTZ: Liu-Tang-Zhang method. LTZ4: modified Liu-Tang-Zhang method. HBE: Hall-Buckley-Eagleson method. } 
\label{fig.tie_box_mu}
\end{figure}

\begin{figure}
\includegraphics[width=1\textwidth]{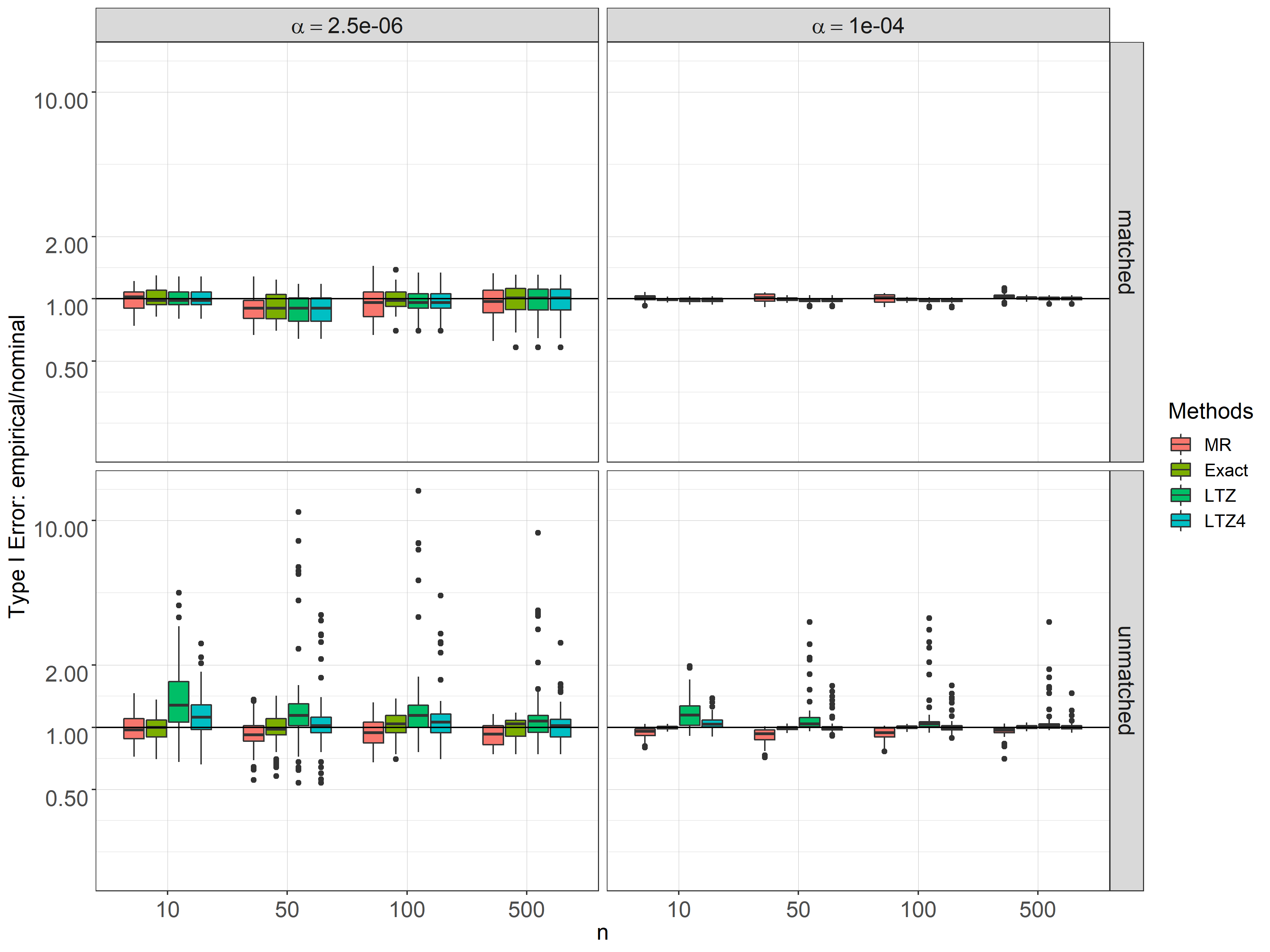}
\caption{Boxplots of the ratios between empirical type I error rates and the nominal levels, $\alpha = 2.5\times 10^{-6}$ and $1\times 10^{-4}$, stratified by whether skewness and kurtosis were matched in the LTZ method.  MR: moment-ratio matching method. Exact: the Davies method. LTZ: Liu-Tang-Zhang method. LTZ4: modified Liu-Tang-Zhang method.} 
\label{fig.tie_box_LTZ}
\end{figure}

\section{Conclusion}
\label{sec.conc}
We propose a moment-ratio matching method that improves the accuracy of approximating the right-tail probability of quadratic forms of Gaussian variables. The proposed method controls the type I error rates almost as good as the exact method while improves the computational efficiency by 1/2 to 2/3. This method is expected to have broad applications in hypothesis testing, such as genetic association studies where Gaussian quadratic forms are frequently used as statistics to combine signals from multiple sources. 

Beyond the Gaussian quadratic forms, we believe the idea of utilizing higher-order moments without exactly matching them, as described in Section~\ref{sec.framework} and \ref{sec.MR}, is applicable to more general distribution approximation problems where matching both skewness and kurtosis is difficult or even impossible.  Our results demonstrate that a wise choice of matching equation that incorporates both skewness and kurtosis, such as the moment-ratio matching, can obtain accurate approximation like both moments are matched.  

%
%
%
%
%

\bibliographystyle{Chicago}

\bibliography{mybibfile}
\end{document}